\documentclass{article}
\usepackage{spconf,amsmath,graphicx,hyperref}
\usepackage{amsmath,amssymb,amsfonts}
\usepackage{algorithmic}
\usepackage{graphicx}
\usepackage{textcomp}
\usepackage{xcolor}
\usepackage{times}
\usepackage{enumerate}
\usepackage{amsmath,amssymb,amsfonts}
\usepackage{mathtools} % for \mathclap

\usepackage{cleveref}
\usepackage{mathrsfs}
\usepackage{comment}
\usepackage{cite}

\def\BibTeX{{\rm B\kern-.05em{\sc i\kern-.025em b}\kern-.08em
    T\kern-.1667em\lower.7ex\hbox{E}\kern-.125emX}}

\usepackage{placeins}
\usepackage{amsthm}

\newtheorem{remark}{Remark}
\newtheorem{lemma}{Lemma}

\newtheorem{proposition}{Proposition}
\newtheorem{theorem}{Theorem}

\title{An Information Geometric Approach to Fairness With Equalized Odds Constraint}
\name{Amirreza Zamani$^\dagger$, Ayfer \"{O}zg\"{u}r$^\ddagger$, Mikael Skoglund$^\dagger$}
\address{$^\dagger$Division of Information Science and Engineering, KTH Royal Institute of Technology \\
			$^\ddagger$Department of Electrical Engineering, Stanford\\
\texttt{amizam@kth.se,  aozgur@stanford.edu, skoglund@kth.se}}
\begin{document}
%\ninept
%
\maketitle
\begin{abstract}
We study the statistical design of a fair mechanism that attains equalized odds, where an agent uses some useful data (database) $X$ to solve a task $T$. Since both $X$ and $T$ are correlated with some latent sensitive attribute $S$, the agent designs a representation $Y$ that satisfies an equalized odds, that is, such that $I(Y;S|T) =0$. In contrast to our previous work, we assume here that the agent has no direct access to $S$ and $T$; hence, the Markov chains $S - X - Y$ and $T - X - Y$ hold. Furthermore, we impose a geometric structure on the conditional distribution $P_{S|Y}$, allowing $Y$ and $S$ to have a small correlation, bounded by a threshold.
When the threshold is small, concepts from information geometry allow us to approximate mutual information and reformulate the fair mechanism design problem as a quadratic program with closed-form solutions under certain constraints. For other cases, we derive simple, low-complexity lower bounds based on the maximum singular value and vector of a matrix. Finally, we compare our designs with the optimal solution in a numerical example.

	%When this threshold is small, concepts from information geometry allow us to locally approximate the mutual information. By utilizing this approximation the main complex fair mechanism design problem can be rewritten as a quadratic optimization problem that has simple closed-form solution under certain constraints. For the cases where the closed-form solution is not obtained we derive simple lower bounds with low computational complexity. Here, we provide simple fairness designs with low complexity which are based on finding the maximum singular value and singular vector of a matrix.
	%Finally, in a numerical example we compare our obtained results with the optimal solution.
\end{abstract}
\begin{keywords}
fairness and privacy design, equalized odds, information geometry, KL-approximation, $\chi^2$ point-wise measure 
\end{keywords}
\vspace{-3mm}
\section{Introduction}
\label{sec:intro}
\vspace{-2mm}
%As shown in Fig.~\ref{ICASSP}, in this paper, an agent wants to use some observable useful data (database) $X$ to draw inferences or make some decision about a task $T$. Here, we assume that both the useful data and the task are correlated with some latent sensitive attribute or secret $S$. For example, the data $X$ could represent a person's resume, the task $T$ could determine whether this person should be employed in a position or not, and the sensitive attribute $S$ could correspond to the person's gender or ethnicity.
As shown in Fig.~\ref{ICASSP}, an agent uses observable data $X$ to make decisions about a task $T$, where both $X$ and $T$ are correlated with a latent sensitive attribute $S$. For example, $X$ could be a resume, $T$ an employment decision, and $S$ the applicant’s gender or ethnicity.

As outlined in \cite{AmirITW2024,vari}, it is important to ensure that decisions are not unfair and that inferences do not violate privacy leakage constraints. To this end, we can design a \emph{private} or \emph{fair} representation $Y$ of the data, that is, a representation that contains the maximum possible information about the task $T$ while satisfying fairness or privacy criteria with respect to the sensitive attribute or secret
\cite{AmirITW2024, vari,zhao2022, zhao2019,zemel,hardt, king3, borz, khodam, Khodam22,Yanina1,fairAsoodeh}. Privacy and fairness are closely related \cite{vari,AmirITW2024,fairAsoodeh} and can be jointly leveraged in applications such as classification. In particular, \cite{AmirITW2024} considers the design of $Y$ that satisfies $I(Y;S) = 0$. %This has different meanings in different communities.
 In the privacy literature, independence of $Y$ and $S$ is called \emph{perfect privacy} or \emph{perfect secrecy}~\cite{borz,Yanina1}, while in the fairness literature it is known as \emph{perfect demographic} (or \emph{statistical}) \emph{parity}~\cite{vari,zhao2022,zhao2019,zemel}.
As noted in \cite{11123370,khodam,Yanina1,fairAsoodeh,king3}, perfect privacy or demographic parity is often unattainable; thus \cite{11123370} proposes a bounded criterion, $I(S;Y)\leq \epsilon$.
 In contrast to \cite{AmirITW2024,11123370}, here, we use the \emph{equalized-odds} notion of fairness~\cite{hardt}, employing $I(Y;S\mid T)=0$ as the criterion. We also assume that neither the sensitive attribute $S$ nor the task $T$ is directly accessible to the agent (Fig.~\ref{ICASSP}), i.e., the Markov chains $S\!-\!X\!-\!Y$ and $T\!-\!X\!-\!Y$ hold.

\begin{figure}[]
	\centering
	\includegraphics[width = 0.39\textwidth]{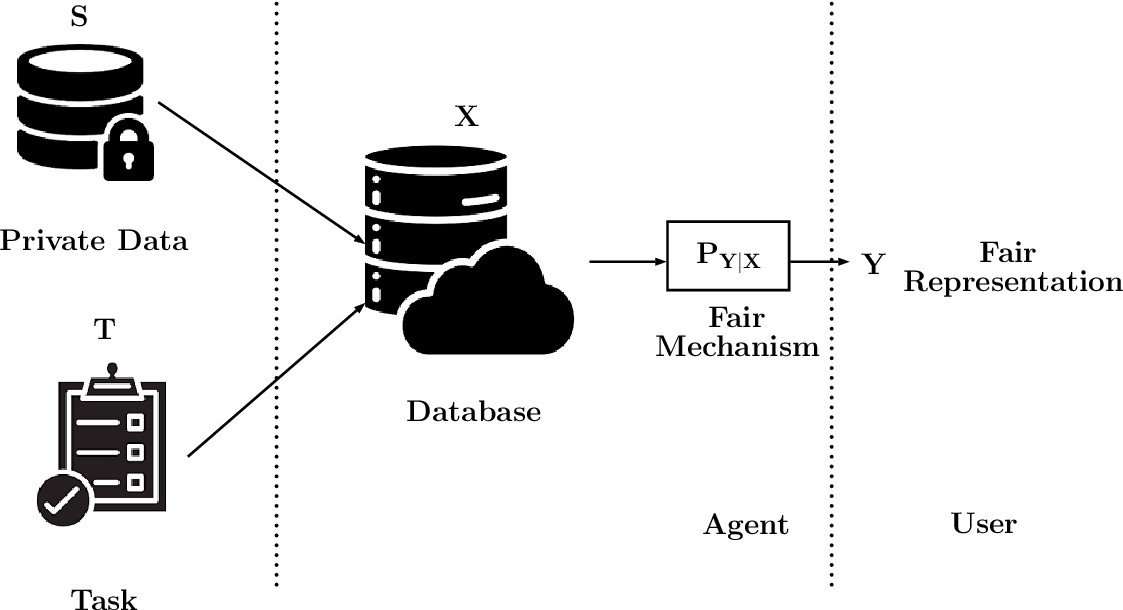}
	\caption{Data representation with an equalized odds constraint. The goal is to design a representation $Y$ of the data $X$ that is useful for the task $T$, is compressed, and leaks within a controlled threshold of the sensitive attribute $S$. Here, the agent has no direct access to $S$ and $T$.}
	\label{ICASSP}
    \vspace{-5mm}
\end{figure}
Many information-theoretic problems are challenging due to the lack of geometric structure on probability distributions. When distributions are close, KL divergence and mutual information can be approximated by a weighted squared Euclidean distance, enabling tractable approximations. This approach has been applied to point-to-point and broadcast channels \cite{Shashi,huang}, as well as to privacy settings \cite{khodam,Khodam22,king3,emma,8125176}.
Under the fairness requirements, we study the trade-off between utility, compression, and equalized odds. Specifically, we maximize the information that $Y$ conveys about the task $T$, subject to the compression constraint $I(Y;X)\le r$ and the equalized-odds condition $I(S;Y|T)=0$. In addition, we impose a point-wise constraint on $P_{S|Y}$, allowing us to use information-geometric tools as in \cite{khodam} to approximate mutual information and reduce the problem to a tractable quadratic linear-algebra form with intuitive insight.
\\
\textbf{Prior works on fairness:}
The notion of \emph{fair representations} was introduced by Zemel et al.~\cite{zemel}, advancing algorithmic fairness through deep learning. Later work has focused on adversarial~\cite{zemel,edwards2016censoring,zhao2019} and variational~\cite{vari,creager2019flexibly,louizos2015variational,gupta2021controllable} approaches, while theoretical trade-offs between utility and fairness are studied in~\cite{zhao2022}.
In \cite{vari}, the authors study an information-theoretic fairness problem and introduce the Conditional Fairness Bottleneck (CFB), which captures the trade-offs among utility, fairness, and compression of representations in terms of mutual information. This criterion has since been adopted in subsequent studies~\cite{gupta2021controllable,de2022funck}, with representation design based on a variational approach. In contrast, \cite{AmirITW2024} proposes a constructive theoretical framework for fair representations with perfect demographic parity, providing upper and lower bounds on the utility--fairness--compression trade-off. The lower bounds are attained using randomization and extensions of the Functional and Strong Functional Representation Lemmas~\cite{king3}. Finally, \cite{11123370} extends this setting by relaxing perfect demographic parity to a bounded statistical-parity constraint.
%Furthermore, they have shown that the lower bounds can be tight under specific assumptions. 
In \cite{fairAsoodeh}, a binary classification problem subject to both differential privacy and fairness constraints is studied. In \cite{9517766}, the role of data processing in the fairness–accuracy trade-off, using equalized odds as the fairness criterion, has been studied.\\
%\vspace{-1mm}
\textbf{Contributions:}
Here, we propose a simple theoretical approach for designing fair representations under equalized odds, addressing the trade-off between utility, fairness, and compression \eqref{problem}. By imposing a point-wise constraint on $P_{S|Y}$, we restrict the correlation between $Y$ and $S$; when this threshold is small, the optimization reduces to a quadratic problem. The resulting design relies on the maximum singular value and vector of a matrix, and we validate it against the optimal solution in a numerical example.
 \vspace{-4mm}
\section{System model} \label{sec:system}
\vspace{-3mm}
In this section, we introduce the problem of designing a compressed representation of data under equalized odds. Let the sensitive data $S$, shared data (useful data) $X$, and task $T$ be discrete random variables (RVs) defined on alphabets $\mathcal{S}$, $\mathcal{X}$, and $\mathcal{T}$, of finite cardinalities $\left| \mathcal{S} \right|
$, $\left| \mathcal{X} \right|$, and $\left| \mathcal{T} \right|
$, respectively. The marginal probability distributions of $S$, $X$, and $T$ are denoted by $P_S \in \mathbb{R}^{|\mathcal{S}|}$, $P_X \in \mathbb{R}^{|\mathcal{X}|}$, and $P_T \in \mathbb{R}^{|\mathcal{T}|}$, respectively. Furthermore, $P_{S,X,T} \in \mathbb{R}^{|\mathcal{S}|\times |\mathcal{X}| \times |\mathcal{T}|}$ denotes the joint distribution of $S$, $X$ and $T$. We assume that $S$, $T$, and $X$ form the Markov chain $S\!-\!T\!-\!X$; that is, the sensitive attribute $S$ and the database $X$ are correlated only through task $T$. For example, $S$ may be a deterministic function of $T$. We will discuss the benefits of this assumption later. Furthermore, we assume that $P_{T|X}$ and $P_{S|T}$ are invertible matrices yielding $|\mathcal{X}|=|\mathcal{S}|=|\mathcal{T}|$.
Let the discrete RV $Y \in \mathcal{Y}$ denote the fair representation. As we outlined before, the Markov chains $S-X-Y$ and $T-X-Y$ hold, as the agent has no direct access to $S$ and $T$.  
The goal is to design a mapping $P_{Y|X} \in \mathbb{R}^{|\mathcal{Y}| \times |\mathcal{X}|}$ that maximizes the information it keeps about the task $T$, while maintaining a minimum level of compression $r$ and satisfying the equalized odds $I(Y;S|T)=0$, i.e., $Y$ should satisfy the Markov chain $S-T-Y$. As we outlined earlier, in this work, we impose a point-wise (strong) $\chi^2$-constraint on $P_{S|Y}$ as follows. Here, the representation $Y$ should also satisfy
\begin{align}
    \chi^2(P_{S|Y=y};P_S)\!=\!\sum_s\left(\frac{P_{S|y}(s)-P_{S}(s)}{P_S(s)}\right)^2\!\!\!\!\leq \epsilon^2,\ \forall y,\label{local1}
\end{align}
where $\|\cdot\|$ corresponds to the Euclidean norm. 
Intuitively, for small $\epsilon$, \eqref{local1} means that the two distributions (vectors) $P_{S|Y=y}$ and $P_S$ are close to each other.
The closeness of $P_{S|Y=y}$ and $P_S$ allows us to locally approximate $I(T;Y)$ and $I(X;Y)$ which leads to an approximation of \eqref{problem}. Furthermore, \eqref{local1} illustrates how $Y$ and $S$ can be correlated. Using equalized odds, conditioned on $T$, $Y$ and $S$ should be independent, however, $Y$ and $S$ can still be correlated, and \eqref{local1} illustrates the correlation threshold. In other words, we restrict attention to the subset of representations $Y$ that satisfy both the equalized odds and the point-wise constraints. For more details on the point-wise measure, see \cite{khodam,Khodam22}.
Hence, the fair representation design problem can be stated as follows:
\begin{align}
    g^{r}_{\epsilon,\chi^2}(P_{S,X,T})
    &\triangleq\sup_{\begin{array}{c} 
	\substack{P_{Y|X}:S-X-Y,\ T-X-Y,\\ I(Y;S|T)=0,\\\chi^2(P_{S|y};P_S)\leq \epsilon^2,\ \forall y,\\ I(X;Y)\leq r }
	\end{array}}I(Y;T).
    \label{problem}
\end{align} 
The constraint $I(Y;S|T)=0$, corresponds to the equalized odds, the point-wise measure describes the correlation between $Y$ and $S$, and $I(X;Y)\leq r$ represents the compression rate constraint.\vspace{-2mm}
\iffalse
The next result shows several equivalences that help us to rewrite the problem \eqref{problem}.
\begin{proposition}\label{prop111}
    We have the following equivalencies.\\
    \textbf{(i)} $S-X-Y$, $T-X-Y$, and $S-T-Y$.\\
    \textbf{(ii)} $S-T-X-Y$.\\
    \textbf{(iii)} $S-T-X$, $T-X-Y$, and $S-X-Y$.
\end{proposition}
\begin{proof}
    The proof is followed by the definition of the Markov chain. From (i) to (ii), it suffices to show that $S-T-X$ holds. Using $S-X-Y$ and $T-X-Y$ we have $P_{S|X}P_{X|Y}=P_{S|Y}$ and $P_{T|X}P_{X|Y}=P_{T|Y}$. Using $S-T-Y$, we have $P_{S|T}P_{T|Y}=P_{S|Y}$. By substituting $P_{S|Y}$ and $P_{T|Y}$, we have
    \begin{align*}
    P_{S|T}P_{T|X}P_{X|Y}=P_{S|X}P_{X|Y},
    \end{align*}
    resulting
    $\left(P_{S|T}P_{T|X}-P_{S|X}\right)P_{X|Y}=0$, which completes the proof. A Similar proof can be used for showing other equivalencies.
\end{proof}
Using Proposition \ref{prop111}, we can rewrite \eqref{problem} as follows.
\begin{align}
    g^{r}_{\epsilon,\chi^2}(P_{S,X,T})
    &=\sup_{\begin{array}{c} 
	\substack{P_{Y|X}:S-T-X-Y,\\\chi^2(P_{S|y};P_S)\leq \epsilon^2,\ \forall y,\\ I(X;Y)\leq r }
	\end{array}}I(Y;T).
    \label{problem2}
\end{align}
\fi
\begin{remark}
We define the utility–fairness–compression trade-off, subject to equalized-odds and compression constraints, as follows.
\begin{align}g^{r}_{\epsilon}(P_{S,X,T})
    &\triangleq\sup_{\begin{array}{c} 
	\substack{P_{Y|X}:S-X-Y,\ T-X-Y,\\I(Y;S|T)=0,\\ I(X;Y)\leq r }
	\end{array}}I(Y;T). \label{prob2}\end{align}
    Then, we have $g^{r}_{\epsilon,\chi^2}(P_{S,X,T})\leq g^{r}_{\epsilon}(P_{S,X,T})$. In other words, since problem \eqref{prob2} is hard to solve, we instead study the lower bound \eqref{problem} that can be addressed using the concepts from information geometry.
\end{remark}
\begin{remark}
	In this work, we assume that $P_{S|T}$ and $P_{T|X}$ are invertible; however, this assumption can be generalized using the techniques in \cite{Khodam22}. Here, we focus on invertible cases for design simplicity and leave the generalization to future work.
    %\vspace{-3mm}
\end{remark}
\section{Main Results}\label{sec:resul}
\vspace{-3mm}
In this section, we first derive an approximation to \eqref{problem}. We then find lower bounds for this approximation and discuss their tightness. %Furthermore, we discuss how the approximated problem can be directly solved.\\
%\textbf{Fair mechanism design: information geometric approach:}
%In this part, we find an approximation on \eqref{problem}. 
To do so,
using \eqref{local1}, we can rewrite the conditional distribution $P_{S|Y=y}$ as a perturbation of $P_S$. Thus, for any $y\in\mathcal{Y}$, we can write $P_{S|Y=y}=P_S+\epsilon\cdot J_{y}$, where $J_y\in\mathbb{R}^\mathcal{S}$ is a perturbation vector that has the following three properties
		\begin{align}
	&\sum_s J_{y}(s)=0,\ \ \forall y,\label{0}	 \\
	&\sum_y P_{Y}(y)J_{y}(s)=0, \ \forall s,\label{sum}	\\
	&\|[\sqrt{P_S}^{-1}]J_y\|_2=\sum_x \frac{J_{y}(s)^2}{P_S(s)}\leq 1 \label{1}.
	\end{align} 
We note that \eqref{sum} implies $\sum_u P_Y(y)J_y=0\in\mathbb{R}^{|\mathcal{S}|}$. 
The first two properties ensure that $P_{S|Y=y}$ is a valid probability distribution \cite{khodam,Khodam22}, and the third property follows from \eqref{local1}.
Next, we approximate $I(Y;X)$ and $I(Y;T)$ using the Markov chains and the perturbation vector $J_y$.
In the next results, we use the Bachmann-Landau notation where $o(\epsilon)$ describes the asymptotic behavior of a function $f:\mathbb{R}^+\rightarrow\mathbb{R}$ which satisfies $\frac{f(\epsilon)}{\epsilon}\rightarrow 0$ as $\epsilon\rightarrow 0$. Here, $[\sqrt{P_X}]$ denotes the diagonal matrix with entries $\sqrt{P_X(x)}$.
\begin{lemma}\label{lem1}
	For all $\epsilon<\frac{\min_{x\in\mathcal{X}}P_X(x)}{|\sigma_{\text{max}}(P_{T|X}^{-1}P_{S|T}^{-1})|\sqrt{\max_{s\in{\mathcal{S}}}P_S(s)}}$, we have
	\begin{align}
I(X;Y)&\!=\!\frac{1}{2}\epsilon^2\sum_y\! P_y\|[\sqrt{P_X}^{-1}]P_{T|X}^{-1}P_{S|T}^{-1}J_y\|^2\!+\!o(\epsilon^2)\nonumber\\&\simeq \frac{1}{2}\epsilon^2\sum_y P_y\|[\sqrt{P_X}^{-1}]P_{T|X}^{-1}P_{S|T}^{-1}J_y\|^2\label{appxy}
\end{align}
and for all $\epsilon<\frac{|\sigma_{\text{min}}(P_{S|T})|\min_{t\in\mathcal{T}}P_T(t)}{\sqrt{\max_{s\in{\mathcal{X}}}P_S(s)}}$, 
    we have
\begin{align}
I(T;Y)&=\frac{1}{2}\epsilon^2\sum_y P_y\|[\sqrt{P_T}^{-1}]P_{S|T}^{-1}J_y\|^2+o(\epsilon^2)\nonumber\\&\simeq \frac{1}{2}\epsilon^2\sum_y P_y\|[\sqrt{P_T}^{-1}]P_{S|T}^{-1}J_y\|^2\label{appty}
	\end{align}
\end{lemma}
%\begin{proof}
\textbf{Proof:}
	To approximate $I(T;Y)$, using the Markov chain $S-T-Y$ and invertible $P_{S|T}$ we can rewrite $P_{T|Y=y}$ as perturbations of $P_T$ as follows:
	%\begin{align*}
	$P_{T|Y=y}=P_{S|T}^{-1}[P_{S|Y=y}-P_S]+P_T=\epsilon\cdot P_{S|T}^{-1}J_y+P_T.$
	%\end{align*}
	Then, by using local approximation of the KL-divergence which is based on the second order Taylor expansion of $\log(1+x)$ we get
	%\begin{align*}
	$I(Y;T)=\sum_y P_Y(y)D(P_{T|Y=y}||P_T)=\sum_y P_Y(y)\times\\\sum_t\! P_{T|Y=y}(t)\log\!\!\left(\!1\!+\!\epsilon\frac{P_{S|T}^{-1}J_y(t)}{P_T(t)}\right)=\frac{1}{2}\epsilon^2\sum_y P_Y(y)\times\\\sum_t
	\frac{(P_{S|T}^{-1}J_y(t))^2}{P_T(t)}+o(\epsilon^2).$
	%&=\frac{1}{2}\epsilon^2\sum_u %P_U\|WL_u\|^2+o(\epsilon^2).
	%\end{align*}
	For more details see \cite[Proposition 3]{khodam}. Due to the Taylor expansion, we must have $|\epsilon\frac{P_{S|T}^{-1}J_y(t)}{P_T(t)}|<1$ for all $t$ and $y$. A sufficient condition for $\epsilon$ to satisfy this inequality is to have $\epsilon<\frac{|\sigma_{\text{min}}(P_{S|T})|\min_{t\in\mathcal{T}}P_T(t)}{\sqrt{\max_{s\in{\mathcal{X}}}P_S(s)}}$, since in this case we have
\begin{align*}
\epsilon^2|P_{S|T}^{-1}J_y(t)|^2&\leq\epsilon^2\left\lVert P_{S|T}^{-1}J_y\right\rVert^2\leq\epsilon^2 \sigma_{\max}^2\left(P_{S|T}^{-1}\right)\left\lVert J_y\right\rVert^2\\&\stackrel{(a)}\leq\frac{\epsilon^2\max_{s\in{\mathcal{S}}}P_S(s)}{\sigma^2_{\text{min}}(P_{S|T})}<\min_{t\in\mathcal{T}} P_T^2(t),
\end{align*}
which implies $|\epsilon\frac{P_{S|T}^{-1}J_y(t)}{P_T(t)}\!|<1$. The step (a) follows from $\sigma_{\max}^2\left(P_{S|T}^{-1}\right)=\frac{1}{\sigma_{\min}^2\left(P_{S|T}\right)}$ and $\|J_y\|^2\leq\max_{s\in{\mathcal{S}}}P_S(s)$. The latter inequality follows from \eqref{1} since we have\\
%\begin{align*}
$\frac{\|J_y\|^2}{\max_{s\in{\mathcal{S}}}P_S(s)}\leq \sum_{s\in\mathcal{S}}\frac{J_y^2(s)}{P_S(s)}\leq 1.$
%\end{align*}
Furthermore, using the Markov chain $T-X-Y$ and invertible $P_{T|X}$ we can rewrite $P_{X|Y=y}$ as perturbations of $P_X$ as follows
	%\begin{align*}
	$P_{X|Y=y}\!=\!P_{T|X}^{-1}[P_{T|Y=y}-P_T]\!+\!P_X\!=\!\epsilon P_{T|X}^{-1}P_{S|T}^{-1}J_y\!+\!P_X.$
	%\end{align*}
    Then, using similar approach we have
    %\begin{align*}
	$I(Y;X)=\sum_y P_Y(y)D(P_{X|Y=y}||P_X)\!\!=\sum_y P_Y(y)\times\\\sum_x\! P_{X|Y=y}(x)\log\!\!\left(\!1\!+\!\epsilon\frac{P_{T|X}^{-1}P_{S|T}^{-1}J_y(t)}{P_X(x)}\right)\!=\frac{1}{2}\epsilon^2\sum_y P_Y(y)\times\\\sum_x
	\frac{(P_{T|X}^{-1}P_{S|T}^{-1}J_y(x))^2}{P_X(x)}+o(\epsilon^2).$
	%&=\frac{1}{2}\epsilon^2\sum_u %P_U\|WL_u\|^2+o(\epsilon^2).
	%\end{align*}
    Finally, for approximating $I(X;Y)$ we must have $|\epsilon\frac{P_{T|X}^{-1}P_{S|T}^{-1}J_y(x)}{P_x(x)}|<1$ for all $x$ and $y$. A sufficient condition for $\epsilon$ to satisfy this inequality is to have $\epsilon<\frac{\min_{x\in\mathcal{X}}P_X(x)}{|\sigma_{\text{max}}(P_{T|X}^{-1}P_{S|T}^{-1})|\sqrt{\max_{s\in{\mathcal{S}}}P_S(s)}}$, since in this case we have
    \begin{align*}
&\!\!\!\!\!\!\epsilon^2|P_{T|X}^{-1}P_{S|T}^{-1}J_y(x)|^2\leq\epsilon^2\left\lVert P_{T|X}^{-1}P_{S|T}^{-1}J_y\right\rVert^2\\&\!\!\!\!\!\!\leq\epsilon^2 \sigma_{\max}^2\left(P_{T|X}^{-1}P_{S|T}^{-1}\right)\left\lVert J_y\right\rVert^2\\&\!\!\!\!\!\!\leq\epsilon^2\max_{s\in{\mathcal{S}}}P_S(s)\sigma^2_{\text{max}}(P_{T|X}^{-1}P_{S|T}^{-1})<\min_{x\in\mathcal{X}} P_X^2(x).
\end{align*}
%\end{proof}
%In the next result, we present an approximation on \eqref{problem}. 
Let $L_y\triangleq[\sqrt{P_S}^{-1}]J_y$, $W^{T;Y}\!\triangleq\![\sqrt{P_T}^{-1}]P_{S|T}^{-1}[\sqrt{P_S}]$, $W^{X;Y}\!\triangleq\![\sqrt{P_X}^{-1}]P_{T|X}^{-1}P_{S|T}^{-1}[\sqrt{P_S}]$,\\ $c_1\triangleq \frac{\min_{x}P_X(x)}{|\sigma_{\text{max}}(P_{T|X}^{-1}P_{S|T}^{-1})|\sqrt{\max_{s}P_S(s)}}$, $c_2\!\triangleq\! \frac{|\sigma_{\text{min}}(P_{S|T})|\min_{t}P_T(t)}{\sqrt{\max_{s}P_S(s)}}$.
\begin{theorem}\label{th1}
    For all $\epsilon<\min\{c_1,c_2\}$, \eqref{problem} can be approximated by the following problem
    \begin{align}\label{approx}
\eqref{problem}\simeq P_2\triangleq\!\!\!\!\!\!\!\!\!\!\!\!\!\!\!\!\max_{\begin{array}{c} 
	\substack{P_{y},L_y:L_y\perp \sqrt{P_S},\\\sum_y P_yL_y=0,\\ \|L_y\|^2\leq 1,\\ \frac{1}{2}\epsilon^2\sum_y P_y\|W^{X;Y}L_y\|^2\leq r}
	\end{array}}\!\!\!\!\!\!\!\frac{1}{2}\epsilon^2\sum_y P_y\|W^{T;Y}L_y\|^2. 
    \end{align}
\end{theorem}
\textbf{Proof:}
	The proof follows by Lemma \ref{lem1}. The constraints $L_y\perp \sqrt{P_S}$, $\sum_y P_yL_y=0$, and $\|L_y\|^2\leq 1$ are followed by \eqref{0}, \eqref{sum}, and \eqref{1}, for more detail see \cite{khodam}.
%\end{proof}
%\begin{remark}
%	As $\epsilon$ decreases, the approximation tightens because the Taylor expansion error term vanishes.
%\end{remark}
\begin{remark}
The compression-rate constraint $I(X;Y)\le r$ can be rewritten as
$
\frac{1}{2}\epsilon^{2}\sum_{y} P_{y}\,\|W^{X;Y} L_{y}\|^{2}\le r - o(\epsilon^{2}).
$
Since the compression threshold $r$ is typically much larger than the correlation parameter $\epsilon$, we neglect the $o(\epsilon^{2})$ term relative to $r$. We can also derive an upper bound on the error term and use it here; we leave this to future work.
\vspace{-2mm}
\end{remark}
Next, we present properties of $W^{X;Y}$ and $W^{T;Y}$ used to find bounds on \eqref{approx}; the proof follows \cite{khodam} and is omitted.
\begin{proposition}\label{pos3}
    Each of the matrices $W^{T;Y}$ and $W^{X;Y}$ has singular value $1$, with corresponding singular vector $\sqrt{P_S}$.\vspace{-2mm}
\end{proposition}
%\textbf{Proof:}
 %   The proof is similar to that in \cite{khodam} and is omitted due to space constraints.\\
%\end{proof}
Next, we find lower bounds on \eqref{approx} and discuss their tightness. To this end, let $\sigma_{\max}^{T;Y}$ and $\sigma_{\max_2}^{T;Y}$ denote the largest and second-largest singular values of the matrix $W^{T;Y}$, with corresponding singular vectors $L_{\sigma}^{W^{T;Y}}$ and $L_{\sigma_2}^{W^{T;Y}}$, respectively.

\begin{theorem}\label{th2}
	If $\sigma_{\max}^{T;Y}>1$, we have
	\begin{align}
	P_2\geq \frac{1}{2}\epsilon^2\left(\frac{\sigma_{\max}^{T;Y}}{K}\right)^2
	\end{align}
	where $1\leq K$ is the smallest constant that satisfies 
	%\begin{align}
	\\$\frac{1}{2}\epsilon^2\|W^{X;Y}L_{\sigma}^{W^{T;Y}}\|^2\leq rK^2.$
	%\end{align}
	If $\sigma_{\max}^{T;Y}=1$, then 
	\begin{align}\label{y}
	P_2\geq \frac{1}{2}\epsilon^2\left(\frac{\sigma_{\max_2}^{T;Y}}{K}\right)^2,
	\end{align}
    where $1\leq K$ is the smallest constant that satisfies 
	%\begin{align}
	\\$\frac{1}{2}\epsilon^2\|W^{X;Y}L_{\sigma_2}^{W^{T;Y}}\|^2\leq rK^2.$
	%Finally, when $|\mathcal{S}|=2$, if $\sigma_{\max}^{T;Y}>1$, then we have $P_2= \frac{1}{2}\epsilon^2\left(\frac{\sigma_{\max}^{T;Y}}{K}\right)^2$, otherwise, $P_2= \frac{1}{2}\epsilon^2\left(\frac{\sigma_{\max_2}^{T;Y}}{K}\right)^2$.
    Finally, for $|\mathcal{S}|=2$, 
$
P_2 = \tfrac{1}{2}\epsilon^2\left(\tfrac{\sigma_{\max}^{T;Y}}{K}\right)^2$ or $ 
P_2 = \tfrac{1}{2}\epsilon^2\left(\tfrac{\sigma_{\max_2}^{T;Y}}{K}\right)^2,
$
depending on $\sigma_{\max}^{T;Y}$.
\end{theorem}
\textbf{Proof:}
To derive the first lower bound, let $\sigma_{\max}^{T;Y}>1$. Using Proposition \ref{pos3}, we have $L_{\sigma}^{W^{T;Y}}\perp \sqrt{P_S}$. Let $Y$ be a uniform binary RV. Furthermore, let $L_1=\frac{L_{\sigma}^{W^{T;Y}}}{K}$ and $L_2=\frac{-L_{\sigma}^{W^{T;Y}}}{K}$. Clearly, $L_1\perp \sqrt{P_S}$ and $L_2\perp \sqrt{P_S}$ and $P_Y(1)L_1+P_Y(2)L_2=0$. Moreover, we have
$
\frac{1}{2}\epsilon^2\!\!\left(\sum_y \!\!P_Y\|W^{T;Y}L_y\|^2\right)=\frac{1}{2}\epsilon^2\left(\frac{\sigma_{\max}^{T;Y}}{K}\right)^2. 
$	
The proof of \eqref{y} is similar. To prove the last argument, note that when $|\mathcal{S}|=2$, the matrix $W^{T;Y}$ has two singular vectors, one of which is $\sqrt{P_S}$. Hence, the only feasible direction is the other singular vector which completes the proof.\\
%\end{proof}
After finding $L_y$ and $P_Y$ attaining the lower bounds in Theorem \ref{th2}, we have
$P_{S|Y=0}=P_S+\epsilon[\sqrt{P_S}]L_1$, 
$P_{S|Y=1}=P_S+\epsilon[\sqrt{P_S}]L_2$,
$P_{T|Y=0}=P_T+\epsilon P_{S|T}^{-1}[\sqrt{P_S}]L_1$,
$P_{T|Y=1}=P_T+\epsilon P_{S|T}^{-1}[\sqrt{P_S}]L_2$,
$P_{X|Y=0}=P_X+\epsilon P_{T|X}^{-1}P_{S|T}^{-1}[\sqrt{P_S}]L_1$,
$P_{X|Y=1}=P_X+\epsilon P_{T|X}^{-1}P_{S|T}^{-1}[\sqrt{P_S}]L_2$,
where $L_1$, $L_2$ and marginal distribution of $Y$ are obtained in Theorem \ref{th2}. %Finally, $P_{STXY}(x,y,u)=P_{X|Y}(x|y)P_{Y|U}(y|u)P_U(u)$.
\begin{remark}
    The Markov chain $S-T-X$ helps us verify that the distributions achieving the lower bounds also satisfy the Markov chain $S-X-Y$. Latter follows since we have
   % \begin{align*}
$P_{S|X}P_{X|Y=0}=P_S+\epsilon P_{S|X}P_{T|X}^{-1}P_{S|T}^{-1}[\sqrt{P_S}]L_1\stackrel{(a)}{=}P_S+\epsilon[\sqrt{P_S}]L_1=P_{S|Y=0},$
    %\end{align*}
    where (a) follows from $P_{S|X}P_{T|X}^{-1}=P_{S|T}$, since the Markov chain $S-T-X$ holds.\vspace{-2mm}
\end{remark}
%\tex{Numerical Example:}
Next, we present an example to evaluate the proposed approach and compare the results with \eqref{problem} and \eqref{prob2}.\\
\textbf{Example:}
	Let $P_{T|X}=\begin{bmatrix}
	\frac{1}{4}       & \frac{2}{5}  \\
	\frac{3}{4}    & \frac{3}{5}
	\end{bmatrix}$
	and $P_X$ be given as $[\frac{1}{4} , \frac{3}{4}]^T$. Furthermore, let $P_{S|T}=\begin{bmatrix}
	\frac{1}{2}       & \frac{1}{5}  \\
	\frac{1}{2}    & \frac{4}{5}
	\end{bmatrix}$. Thus, we find $P_T\!=\![0.3625, 0.6375]^T$ and $P_S=P_{S|T}P_T=[0.3088, 0.6913]^T$. Furthermore we find $W^{T;Y}=\begin{bmatrix}
	2.4610    & -0.9206  \\
	-1.1599    &  1.7355 
	\end{bmatrix}$ and $W^{X;Y}=\begin{bmatrix}
	-16.7931    & 11.8246  \\
	-10.3371    &  -5.8669 
	\end{bmatrix}.$
	%\begin{align*}
	%P_X&=P_{X|Y}P_Y=[0.3625, 0.6375]^T,\\
	%W &= [\sqrt{P_Y}^{-1}]P_{X|Y}^{-1}[\sqrt{P_X}] = \begin{bmatrix}
	%-4.8166       & 4.2583  \\
	%3.4761    & -1.5366
	%\end{bmatrix}.
	%\end{align*}
	The singular values of $W^{T;Y}$ are $3.2034$ and $1$ with corresponding right singular vectors $[-0.8314, 0.5557]^T$ and $[0.5557 , 0.8314]^T$, respectively.  Moreover, the
    singular values of $W^{X;Y}$ are $23.7087$ and $1$ with corresponding right singular vectors $[0.8314, -0.5557]^T$ and $[0.5557 , 0.8314]^T$, respectively. Here, we let $\epsilon\in [0.005, 0.05]$ and $r=0.2$.
     $L_1$ equals to $[-0.8314, 0.5557]^T$ and we can check the the compression rate constraint is satisfied for $L_1$ which yields $K=1$. Using Theorem 2, we have $P_2=\frac{1}{2}\epsilon^2(3.2034)^2$ and the lower bound is tight since $|\mathcal{S}|=2$.
In Fig. \ref{geo4}, we compare $P_2$ and the optimal solutions of \eqref{problem} and \eqref{prob2} using exhaustive search. We recall that $P_2$ is the approximate of \eqref{problem} and here achieves the lower bound in Theorem 2. We can see that $g^{r}_{\epsilon}(P_{S,X,T})$ dominates $g^{r}_{\epsilon,\chi^2}(P_{S,X,T})$, and the gap between exact values of \eqref{problem} and \eqref{prob2} is decreasing as $\epsilon$ increases. Finally, the gap between $P_2$ and \eqref{problem} is small in the high privacy regimes. Intuitively, the blue curve dominates \eqref{problem} (red curve), since, compared to \eqref{prob2}, \eqref{problem} imposes the point-wise $\chi^2$-criterion in addition to the equalized odds.\vspace{-4mm}
\begin{figure}[]
	\centering
	\includegraphics[width = 0.48\textwidth]{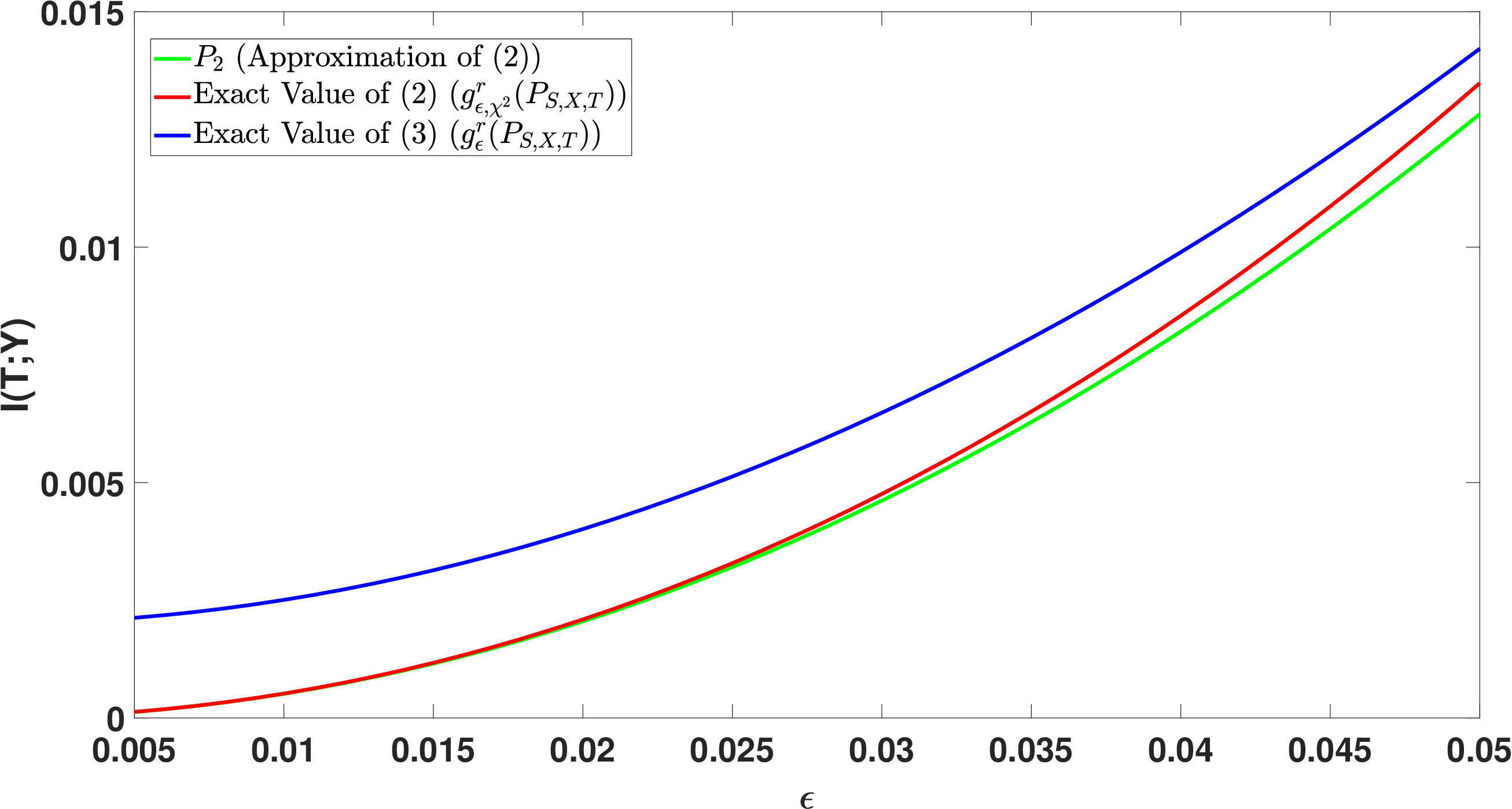}
	\caption{Comparing our method with the optimal solutions of \eqref{problem} and \eqref{prob2}. In high-privacy regimes, $P_2$ is close to \eqref{problem} (via exhaustive search), and $g^{r}_{\epsilon}(P_{S,X,T})$ dominates $g^{r}_{\epsilon,\chi^2}(P_{S,X,T})$.
}
	\label{geo4}
    \vspace{-5mm}
\end{figure}
%\end{example}
\section{conclusion}\label{concul}
\vspace{-3mm}
We have shown that information geometry can be used to simplify an information-theoretic fair mechanism design problem with equalized odds as the fairness constraint. When a small $\epsilon$ correlation threshold is allowed, simple approximate solutions are derived. Specifically, we look for vectors satisfying the fairness constraints of having the largest Euclidean norm, leading to finding the largest principle singular value and vector of a matrix. The proposed approach establishes a useful and general design framework, which has been used in other privacy design problems such as considering point-wise $\chi^2$ and $\ell_1$ criterion in the literature.

\label{sec:refs}

%List and number all bibliographical references at the end of the
%paper. The references can be numbered in alphabetic order or in
%order of appearance in the document. When referring to them in
%the text, type the corresponding reference number in square
%brackets as shown at the end of this sentence \cite{C2}. An
%additional final page (the fifth page, in most cases) is
%allowed, but must contain only references to the prior
%literature.

%Please follow the IEEE Citation Guidelines, \url{https://ieee-dataport.org/sites/default/files/analysis/27/IEEE\%20Citation\%20Guidelines.pdf} for formatting of references.

% References should be produced using the bibtex program from suitable
% BiBTeX files (here: strings, refs, manuals). The IEEEbib.bst bibliography
% style file from IEEE produces unsorted bibliography list.
% -------------------------------------------------------------------------
\clearpage
\bibliographystyle{IEEEtran}
 \bibliography{IEEEabrv,ICASSP2026}

\end{document}